\documentclass[a4paper,11pt]{article}
\usepackage{amsmath, amssymb}
\usepackage{bbm}
\usepackage{caption}
\usepackage{subcaption}
\usepackage{graphicx}
\usepackage{color}
\usepackage{graphicx}
\usepackage{url}
\usepackage{color,soul}
\usepackage{float}
\usepackage{amsmath}
\usepackage{amssymb}
\usepackage{courier}
\usepackage{multicol}
\usepackage[abs]{overpic}
\usepackage{enumitem}
\usepackage{multicol}
\usepackage{amsthm}
\newtheorem{theorem}{Theorem}
\newtheorem{lemma}[theorem]{Lemma}
\newtheorem{claim}[theorem]{Claim}
\newtheorem{proposition}[theorem]{Proposition}

\usepackage[linesnumbered,boxed]{algorithm2e}
\allowdisplaybreaks
\usepackage{pgf,tikz}
\usepackage{mathrsfs}
\usetikzlibrary{arrows}

\def\RR{\mathbb R}
\def\QQ{\mathbb Q}
\def\ZZ{\mathbb Z}

\def\NN{\mathbb N}

\def\bbone{\mathbbm{1}}

\def\cF{\mathcal F}

\def\cH{\mathcal H}

\def\cL{\mathcal L}

\def\cP{\mathcal P}
\def\cT{\mathcal T}
\def\cQ{\mathcal Q}
\def\cR{\mathcal R}

\def\cH{\mathcal H}

\def\cV{\mathcal V}

\def\bdelta{\boldsymbol \delta}

\def\b1{\mathbf 1}

\def\eps{\varepsilon}

\def\sees{\sim}
\newcommand{\raf}[1]{(\ref{#1})}

\newcommand{\poly}{\operatorname{poly}}
\newcommand{\polylog}{\operatorname{polylog}}

\newcommand{\argmax}{\operatorname{argmax}}
\newcommand{\conv}{\operatorname{conv.hull}}
\newcommand{\area}{\operatorname{area}}
\newcommand{\prem}{\operatorname{prem}}
\newcommand{\vol}{\operatorname{vol}}
\newcommand{\cells}{\operatorname{cells}}

\newcommand{\bone}{\ensuremath{\boldsymbol{1}}}
\newcommand{\hide}[1]{}

\def \VCd {\text{VC-dim}}

\def \MO {\textsc{Max}}
\def \PO {\textsc{PointIn}}
\def \Sample {\textsc{Sample}}
\def \SO {\textsc{Subsys}}

\def \OPT {\textsc{Opt}}

\newtheorem{corollary}{Corollary}

\newtheorem{remark}{Remark}

\usepackage[top=1in, bottom=1in, left=1in, right=1in]{geometry}

\title{Finding Small Hitting Sets in Infinite Range Spaces of Bounded VC-dimension}

\author{
Khaled Elbassioni\thanks{Masdar Institute of Science and Technology, P.O. Box 54224, Abu Dhabi, UAE;
(kelbassioni@masdar.ac.ae)}
}

\begin{document}

\maketitle

\begin{abstract}
We consider the problem of finding a small hitting set in an {\it infinite} range space $\cF=(Q,\cR)$ of bounded VC-dimension. We show that, under reasonably general assumptions, the infinite dimensional convex relaxation can be solved (approximately) efficiently by multiplicative weight updates. As a consequence, we get an algorithm that finds, for any $\delta>0$, a set of size $O(s_{\cF}(z^*_\cF))$ that hits $(1-\delta)$-fraction of $\cR$ (with respect to a given measure) in time proportional to $\log(\frac{1}{\delta})$, where  $s_{\cF}(\frac{1}{\epsilon})$ is the size of the smallest $\epsilon$-net the range space admits, and $z^*_{\cF}$ is the value of the {\it fractional} optimal solution. 
This {\it exponentially} improves upon previous results which achieve the same approximation guarantees with running time proportional to $\poly(\frac{1}{\delta})$.
Our assumptions hold, for instance, in the case when the range space represents the {\it visibility} regions of a polygon in $\RR^2$, giving thus a deterministic  polynomial time $O(\log z^*_{\cF})$-approximation algorithm for guarding $(1-\delta)$-fraction of the area of any given simple polygon, with running time proportional to $\polylog(\frac{1}{\delta})$.
\end{abstract}
\section{Introduction}
Let $\cF=(Q,\cR)$ be a range space defined by a set of ranges $\cR\subseteq 2^Q$ over a (possibly) {\it infinite} set $Q$. A {\it hitting} set of $\cR$ is a subset $H\subseteq Q$ such that $H\cap R\ne\emptyset$ for all $R\in\cR$. Finding a hitting set of minimum size for a given range space is a fundamental problem in computational geometry. For finite range spaces (that is when $Q$ is finite), standard algorithms for \textsc{SetCover} \cite{J74,L75,C79} yield $(\log|Q|+1)$-approximation
in polynomial time, and this is essentially the best possible guarantee  assuming $\text{\it NP}\not\subset\text{\it DTIME}(n^{O(\log\log n)})$ \cite{F98,LY94}. Better approximation algorithms exist for special cases, such as range spaces of {\it bounded VC-dimension} \cite{BG95}, of {\it bounded union complexity}~\cite{ClarksonV06,V09}, of {\it bounded shallow cell complexity} \cite{CGKS12}, as well as several classes of geometric range spaces \cite{AES10,PR08,KK11}. Many of these results are based on showing the existence of a small-size {\it  $\epsilon$-net} for the range space $\cF$ and then using the multiplicative weight updates algorithm of  Br\"{o}nnimann and Goodrich \cite{BG95}. For instance, if a range space $\cF$ has VC-dimension $d$ then it admits an $\epsilon$-net of size $O(\frac{d}{\epsilon}\log\frac{1}{\epsilon})$ \cite{HW87,KPW92}, which by the above mentioned method implies an $O(d\cdot\log\OPT_{\cF})$-approximation algorithm for the hitting set problem for $\cF$, where $\OPT_{\cF}$ denotes the size of a minimum-size hitting set. Even et al.~\cite{ERS05} observed that this can be improved to $O(d\cdot\log z^*_{\cF})$-approximation by first solving the LP-relaxation of the problem to obtain the value of the {\it fractional} optimal solution\footnote{In fact, we will observe below (see Appendix~\ref{sec:BG}) that the exact same algorithm of \cite{BG95}, but with a slightly modified analysis, gives this improved bound of \cite{ERS05}, {\it without} the need to solve an LP.}  $z^*_{\cF}$, and then finding an $\epsilon$-net, with $\epsilon:=1/z^*_{\cF}$. 

The multiplicative weight updates algorithm in \cite{BG95} works by maintaining weights on the {\it points}. The straightforward extension to infinite (or continuous) range spaces (that is, the case when $Q$ is infinite) does not seem to work, since the bound on the number of iterations depends on the measure of the regions created during the course of the algorithm, which can be arbitrarily small (see Appendix~\ref{sec:BG} for details). In this paper we take a different approach, which can be thought of as a combination of the methods in \cite{BG95} and \cite{ERS05} (with LP replaced by an {\it infinite dimensional convex relaxation}):
\begin{itemize}
	\item We maintain weights on the { \it ranges} (in contrast to Br\"{o}nnimann and Goodrich \cite{BG95} which maintain weights on the points, and the second method suggested by Agarwal and Pan \cite{AP14} which maintains weights on both points and ranges);
	\item We first solve the covering convex relaxation within a factor of $1+\eps$ using multiplicative weight updates (MWU), extending the approach in \cite{GK98} to infinite dimensional covering LP's (under reasonable assumptions);
	\item We finally use the rounding idea of \cite{ERS05} to get a small integral hitting set from the obtained fractional solution.
\end{itemize}

\paragraph{Informal main theorem.} Given a range space $\cF=(Q,\cR)$ of VC-dimension $d$, (under mild assumptions) there is an algorithm that, for any $\delta>0$, finds a subset of $Q$ of size $O(d\cdot z_{\cF}^*\log z_{\cF}^*)$ that hits $(1-\delta)$-fraction of $\cR$ (with respect to a given measure) in time polynomial in the input description of $\cF$ and $\log(\frac{1}{\delta})$.  

\medskip

This {\it exponentially} improves upon previous results\footnote{More precisely (as pointed to us by an anonymous reviewer), using relative approximation results (see, e.g., \cite{PS11}), one can obtain the same approximation guarantees as our main Theorem by solving the problem on the set system induced on samples of size $O((d\cdot\OPT_{\cF} /\delta) \log (1/\delta))$ .} which achieve the same approximation guarantees, but with running time depending {\it polynomially} on $\frac{1}{\delta}$.  

\medskip

We apply this result to a number of problems:
\begin{itemize}
	\item The art gallery problem: given a simple polygon $H$, our main theorem implies that there is a deterministic  polytime $O(\log z^*_{\cF})$-approximation algorithm (with running time proportional to $\polylog(\frac{1}{\delta})$) for guarding $(1-\delta)$-fraction of the area of $H$. When $\delta$ is (exponentially) small, this improves upon a previous result \cite{CEH07} which gives a  polytime algorithm that finds a set of size $O(\OPT_\cF\cdot\log\frac{1}{\delta})$ hitting $(1-\delta)$-fraction of $\cR$. Other (randomized) $O(\log \OPT_{\cF})$-approximation results which provide full guarding (i.e. $\delta=0$) also exist, but they either run in pseudo-polynomial time \cite{DKDS07}, restrict the set of candidate guard locations \cite{EH06}, or make some general position assumptions \cite{BM16}.
	\item Covering a polygonal region by translates of a convex polygon: Given a collection of polygons in the plane $\cH$  and a convex polygon $H_0$, our main theorem implies that there is a randomized polytime $O(1)$-approximation algorithm for covering  $(1-\delta)$ of the total area of the polygons in $\cH$ by  the minimum number of translates of $H_0$. Previous results with proved approximation guarantees mostly consider only the case when $\cH$ is a set of points \cite{ClarksonV06,HM85,Laue08}.
	\item Polyhedral separation in fixed dimension: Given two convex polytopes $\cP_1,\cP_2\subseteq \RR^d$ such that $\cP_1\subset \cP_2$, our main theorem implies that there is a randomized polytime $O(d\cdot \log z^*_{\cF})$-approximation algorithm for finding a polytope $\cP_3$ with the minimum number of facets separating $\cP_1$ from $(1-\delta)$-fraction of the volume of $\partial\cP_2$. This improves the approximation ratio by a factor of $d$ over the previous (deterministic) result \cite{BG95} (but which gives a complete separation).  
\end{itemize}
More related work on these problems can be found in the corresponding subsections of Section~\ref{sec:app}.

The paper is organized as follows. In the next section we define our notation, recall some preliminaries, and describe the infinite dimensional convex relaxation. In Section~\ref{sec:main}, we state our main result, followed by the algorithm for solving the fractional problem in Section~\ref{sec:algorithm} and its analysis in Section~\ref{sec:analysis}. The success of the whole algorithm relies crucially on being able to efficiently implement the so-called {\it maximization oracle}, which essentially calls for finding, for a given measure on the ranges, a point that is contained in the heaviest subset of ranges (with respect to the given measure). We utilize the fact that the dual range space has bounded VC-dimension in section~\ref{sec:max} to give an efficient randomized implementation of the maximization oracle in the {\it Real RAM} model of computation. With more work, we show in fact that, in the case of the art gallery problem, the maximization oracle can be implemented in deterministic polynomial time in the {\it bit model}; this will be explained in Section~\ref{sec:gallery}. Sections~\ref{sec:cover-polygon} and~\ref{sec:poly-sep}  describe the two other applications.            

\section{Preliminaries}\label{sec:prelim} 
\subsection{Notation}
Let $\cF=(Q,\cR)$ be a range space. The {\it dual} range space $\cF^*=(Q^*,\cR^*)$ is defined as the range space with $Q^*:=\cR$ and $\cR^*:=\{\{R\in\cR:~q\in\cR\}:~q\in Q\}$. For a point $q\in Q$ and a subset of ranges $\cR'\subseteq\cR$, let $\cR'[q]:=\{R\in\cR':~q\in R\}$. 
For a set of points $P\subseteq Q$, let $\cR|_P:=\{R\cap P~:~R\in\cR\}$ be the {\it projection} of $\cR$ onto $P$. Similarly, for a set of ranges $\cR'\subseteq\cR$, let $\cQ_{\cR'}:=\{\cR'[q]:~q\in Q\}$. 
For a finite set $P\subseteq Q$ of size $r$, we denote by $g_{\cF}(r)\le 2^r$ the smallest integer such that $|\cR_P|\le g_{\cF}(r)$.
For $p\in Q$ and $R\in\cR$, we denote by $\bbone_{p\in R}\in\{0,1\}$ the indicator variable that takes value $1$ if and only if $p\in R$.
 
\subsection{Problem definition and assumptions}
More formally, we consider the following problem: 
\begin{description}
	\item \textsc{Min-Hitting-Set}: Given a range space $\cF=(Q,\cR)$, find a minimum-size hitting set.
\end{description}

We shall make the following assumptions\footnote{For simplicity of presentation, we will make the implicit assumption in this paper that both $Q$ and $\cR$ is are in one-to-one correspondence with some subsets of $\RR^k$, as all the applications we consider have this restriction. 
	This implies that the measure $w_0$ in (A3) (and $\mu_0$ in (A3$'$)) can be taken as the standard volume measures in $\RR^k$, and the integrals used below are the standard Riemann integrals. However, we note that the extension to general measurable sets should be straightforward.}:  
\begin{itemize}
	\item[(A1)] $g_\cF(r)\le r^\gamma$, for some non-decreasing function $g:\NN\to\RR_+$, and some constant $\gamma\ge 1$.
	 \item[(A1$'$)] The range space is given by a {\it subsystem oracle} \SO$(\cF,P)$ that, given any finite $P\subseteq Q$, returns the set of ranges $\cR|_P$. 
    \item[(A2)] There exists a finite integral optimum whose value $\OPT_{\cF}$ is bounded by a parameter $n$ (that is not necessarily part of the input).
	
	\item[(A3)] There exists a finite measure $w_0:\cR\to\RR_+$ such that all subsets of $\cR$ are $w_0$-measurable. 
\end{itemize}                                                                                                                                                                  
\subsection{Range spaces of bounded VC-dimension}
We consider range spaces of bounded {\it VC-dimension} defined as follows. 
A finite set $P\subseteq Q$ is said to be {\it shattered} by $\cF$ if $\cR|_P=2^P$. The VC-dimension of $\cF$, denoted $\VCd(\cF)$, is the cardinality of the largest subset of $Q$ shattered by $\cF$. If arbitrarily large subsets of $Q$ can be shattered then the $\VCd(\cF)=+\infty$. It is well-known that if $\VCd(\cF)=d$ then $g_\cF(r)\le O(r^d)$. More precisely, the following bound holds.

\begin{lemma}[Sauer–-Shelah Lemma \cite{Sa72,Sh72}]\label{VC}
	For any range space $\cF=(Q,\cR)$ of VC-dimension $d$ and any $r\ge 1$, it holds that $g_{\cF}(r)\le g(r,d):=\sum_{i=0}^{d}\binom{r}{i}$.
\end{lemma}
\begin{lemma}
	\label{VC2}
	If $\VCd(\cF)=d$ then $\VCd(\cF^*)<2^{d+1}$.
\end{lemma}

%
%
\subsection{$\epsilon$-nets}
Given a range space $(Q,\cR)$, a finite measure $\mu:Q\to\RR_+$ (such that the ranges in $\cR$ are $\mu$-measurable), and 
a parameter $\epsilon>0$, an {\it $\epsilon$-net} for $\cR$ (w.r.t. $\mu$) is a set $P\subseteq Q$ such that $P\cap R\neq\emptyset$ for all $R\in\cR$ that satisfy $\mu(R)\ge\epsilon\cdot\mu(Q)$. 
We say that a range space $\cF$ admits an $\epsilon$-net of size $s_{\cF}(\cdot)$, if for any $\epsilon>0$, there an $\epsilon$-net of size $s_{\cF}(\frac{1}{\epsilon})$. For range spaces of VC-dimension $d$, Haussler and Welzl \cite{HW87} proved a bound of $s_{\cF}(\frac{1}{\epsilon}):=O(\frac{d}{\epsilon}\log\frac{d}{\epsilon})$ on the size of an $\epsilon$-net, which was later slightly improved by Koml\'{o}s et al. 
 
\begin{theorem}[$\epsilon$-net Theorem \cite{HW87,KPW92}]\label{rand-net}
Let $\cF=(Q,\cR)$ be a range space of VC-dimension $d$, $\mu$ be an arbitrary probability measure on $Q$ (such that the ranges in $\cR$ are $\mu$-measurable), and $\epsilon>0$ be a given parameter. Then there exists an $\epsilon$-net of size $s_{\cF}(\frac{1}{\epsilon})=O(\frac{d}{\epsilon}\log\frac{1}{\epsilon})$. In fact, a random sample (w.r.t. to the probability measure $\mu$) of size $s_{\cF}(\frac{1}{\epsilon})$ is an $\epsilon$-net with (high) probability $\Omega(1)$. 
\end{theorem} 

We say that a finite measure $\mu:Q\to\RR_+$ has (finite) support $K$ if $\mu$ can be written as a conic combination of $K$ {\it Dirac measures}\footnote{The Dirac measure satisfies $\int_{Q'}\bdelta_p(q)dq=1$ if $p\in Q'$, and $\int_{Q'}\bdelta_p(q)dq=0$ otherwise.}: $\mu=\sum_{p\in P}\mu(p)\bdelta_p(q)$, for some finite $P\subseteq Q$ of cardinality $K$ and non-negative multipliers $\mu(p)$, for $p\in P$. Measures of finite support can be considered as weights on a finite subset of $Q$, in which case an $\epsilon$-net can be computed deterministically as given by the following result of Matou\v{s}ek \cite{M91}.

\begin{theorem}[\cite{BCM99,CM96,M91}]\label{det-net}
Let $\cF=(Q,\cR)$ be a range space of VC-dimension $d$ satisfying (A1$'$), $\mu$ be a measure on $Q$ with support $K$, and $\epsilon>0$ be a given a parameter. Then for any $\epsilon>0$, there is a deterministic algorithm that computes an $\epsilon$-net for $\cR$ of size $s_{\cF}(\frac{1}{\epsilon})=O(\frac{d}{\epsilon}\log\frac{d}{\epsilon})$ in time $O(d)^{3d}\frac{1}{\epsilon^{2d}}\log^d(\frac{d}{\epsilon})K$. 	
\end{theorem}
Since most of the results on $\epsilon$-nets are stated in terms of the unweighted case, it is worth recalling the reduction from  
the weighted case  to the unweighted case (see , e.g., \cite{M91}). Given a measure $\mu$ defined on a finite set $P$ of support $K$, we replace each point $p\in P$, by $\left\lfloor\frac{\mu(p)K}{\sum_{p\in P}\mu(p)}+1\right\rfloor$ copies of $p$. Let $Q'$ be the new set of points. Then $K':=|Q'|\le 2K$ and an $\frac{\epsilon}{2}$-net for $(Q',\cR|_{Q'})$ is an $\epsilon$-net for $(Q,\cR|_Q)$.

It should also be noted that some special range spaces may admit a smaller size $\epsilon$-net, e.g., $s_{\cF}(\frac{1}{\epsilon})=O(\frac{1}{\epsilon})$ for half-spaces in $\RR^3$ \cite{MSW90,M92}; see also \cite{CGKS12,KK11,KV06,V09}.


\subsection{$\epsilon$-approximations}
Given the dual range space $\cF^*$, a measure $w:\cR\to\RR_+$, and an $\epsilon>0$, an $\epsilon$-approximation is a finite subset of ranges $\cR'\subseteq\cR$ such that, for all $q\in Q$,   
\begin{equation}\label{eps-approx}
\left|\frac{|\cR'[q]|}{|\cR'|}-\frac{w(\cR[q])}{w(\cR)}\right|\le\epsilon;
\end{equation}
see, e.g., \cite{C00}. The following theorem (stated in the dual space for our purposes, where $\VCd(\cF^*)< 2^{d+1}$ by Lemma~\ref{VC2}) states the existence of an $\epsilon$-approximation of small size. 
\begin{theorem}[$\epsilon$-approximation Theorem \cite{AS08,C00,VC71}]\label{thm:eps-approx}
Let $\cF=(Q,\cR)$ be a range space of VC-dimension $d$, $w$ be an arbitrary probability measure on $\cR$, and $\epsilon>0$ be a given a parameter. Then a random sample (w.r.t. the probability measure $w$) of size $O(\frac{d2^{d}}{\epsilon^2}\log\frac{1}{\epsilon\sigma})$ is an $\epsilon$-approximation for $\cF^*$, with probability $1-\sigma$. 
\end{theorem}  
\subsection{The fractional problem}
Given a range space $\cF=(Q,\cR)$, satisfying assumptions (A1)-(A3), the fractional problem seeks to find a measure $\mu$ on $Q$, such that $\mu(R)\ge 1$ for all $R\in \cR$ and $\mu(Q)$ is minimized\footnote{We may as well restrict $\mu$ to have finite support and replace the integrals over $Q$ by summations.}:

\begin{align}\label{FH}
z^*_{\cF}:=\tag{\textsc{F-hitting}}\inf_{\mu} & \int_{q \in Q} \mu(q)dq\\
\text{s.t.} \hspace{6pt} & \int_{q \in R}\mu(q)dq \geq 1, \forall R \in \cR, \label{e-1}\\
& \mu(q) \geq 0, \forall q \in Q.\nonumber
\end{align}
Equivalently, it is required to find a {\it probability} measure $\mu:Q\to[0,1]$ that solves the maximin problem: $\sup_{\mu}\inf_{R\in\cR}\mu(R)$.

\medskip

\begin{proposition}\label{p1} 
	For a range space $\cF$ satisfying (A2), we have $\OPT_{\cF}\ge z^*_{\cF}$.
\end{proposition}
\begin{proof}
	Given a finite integral optimal solution $P^*$, we define a measure $\mu$ of support $\OPT_{\cF}$ by  $\mu(q):=\sum_{p\in P^*}\bdelta_p(q)$. Then $\mu(Q)=\int_{q\in Q}\sum_{p\in P^*}\bdelta_p(q)dq=\sum_{p\in P^*}\int_{q\in Q}\bdelta_p(q)dq=\sum_{p\in P^*}1=|P^*|=\OPT_{\cF}$ and $\mu(R)=\int_{q\in R}\sum_{p\in P^*}\bdelta_p(q)dq=\sum_{p\in P^*}\int_{q\in R}\bdelta_p(q)dq=\sum_{p\in P^*}\bbone_{p\in R}=|\{p\in P^*:~p\in R\}|\ge 1$, for all $R\in\cR$, since $P^*$ is a hitting set. Since $\mu$ is feasible for \raf{FH}, the claim follows.
\end{proof} 

Assume $\cF$ satisfies (A3). For $\alpha\ge 1$, we say that $\mu:Q\to\RR_+$ is an {\it $\alpha$-approximate} solution for \raf{FH} if $\mu$ is feasible for \raf{FH} and $\mu(Q)\leq\alpha \cdot z^*_{\cF}.$
For $\beta\in[0,1]$, we say that $\mu$ is $\beta$-feasible if $\mu(R)\ge 1$ for all $R\in\cR'$, where $\cR'\subseteq\cR$ satisfies $w_0(\cR')\ge\beta \cdot w_0(\cR)$. Finally, we say that $\mu$ is an $(\alpha,\beta)$-approximate solution for \raf{FH} if $\mu$ is $\alpha$-approximate and $\beta$-feasible. 
\subsection{Rounding the fractional solution}

Br\"{o}nnimann and Goodrich \cite{BG95} gave a multiplicative weight updates algorithm for approximating the minimum hitting set for a {\it finite} range space satisfying (A1$'$) and admitting an $\epsilon$-net of size $s_{\cF}(\frac{1}{\epsilon})$. For completeness, their algorithm is given as Algorithm~\ref{BG-alg} in Appendix~\ref{sec:BG}, and works as follows. It first guesses the value of the optimal solution (within a factor of 2), and initializes the weights of all {\it points} to $1$. It then invokes Theorem \ref{det-net} to find an $\epsilon=\frac{1}{2\OPT_{\cF}}$-net of size $s_{\cF}(\frac{1}{\epsilon})$. If there is a range $R$ that is not hit by the net (which can be checked by the subsystem oracle), the weights of all the points in $R$ are doubled. The process is shown to terminate in $O(\OPT_{\cF}\log\frac{|Q|}{\OPT_{\cF}})$ iterations, giving an $s_{\cF}(2\OPT_{\cF})/\OPT_{\cF}$-approximation.
Even et al. \cite{ERS05} strengthen this result by using the linear programming relaxation to get $s_{\cF}(z_{\cF}^*)/z^*_{\cF}$
-approximation. We can restate this result as follows. 

\begin{lemma}\label{l111}
	Let $\cF=(Q,\cR)$ be a range space admitting an $\epsilon$-net of size $s_{\cF}(\frac{1}{\epsilon})$ and $\mu$ be a measure on $Q$ satisfying \raf{e-1}. Then there is a hitting set for $\cR$ of size $s_{\cF}(\mu(Q))$. 
	\end{lemma} 

\begin{proof}
	Let $\epsilon:=\frac{1}{\mu(Q)}$. Then for all $R\in\cR$ we have $\mu(R)\ge1=\epsilon\cdot\mu(Q)$, and hence an $\epsilon$-net for $\cR$ is actually a hitting set.   
\end{proof}
\begin{corollary}\label{cor1}
	Let $\cF=(Q,\cR)$ be a range space of VC-dimension $d$ and $\mu$ be a measure on $Q$ satisfying \raf{e-1}. Then a random sample of size $O(d\cdot\mu(Q)\log(\mu(Q)))$, w.r.t. the probability measure $\mu':=\frac{\mu}{\mu(Q)}$, is a hitting set for $\cR$  with probability $\Omega(1)$. 
	Furthermore, if $\mu$ has support $K$ then there is a deterministic algorithm that computes a hitting set for $\cR$ of size $O(d\cdot\mu(Q)\log(d\cdot\mu(Q)))$ in time $O(d)^{3d}\mu(Q)^{2d}\log^d(d\cdot\mu(Q))K$. 
	\end{corollary} 
\begin{proof}
	In view of Lemma~\ref{l111}, the two parts of the corollary follow from Theorems~\ref{rand-net} and \ref{det-net}, respectively.
\end{proof}


 Further improvements on the Br\"{o}nnimann-Goodrich algorithm can be found in \cite{AP14}.

\section{Solving the fractional problem -- Main result}\label{sec:main}

We make the following further assumption:
\begin{itemize}
\item[(A4)] There is a deterministic (resp., randomized) oracle \MO$(\cF,w,\omega)$ (resp., \MO$(\cF,w,\sigma,\omega)$), that given a range space $\cF=(Q,\cR)$, a finite measure $w:\cR\to\RR_+$ on $\cR$, and $\omega>0$, returns (resp., with probability $1-\sigma$) a point $p\in Q$ such that
   $$
   \xi_w(p)\geq (1-\omega)\max_{q\in Q}\xi_w(q),
   $$
   where $\xi_w(p):=w(\cR[p])=\int_{R\in\cR}w(R)\bbone_{p\in R}dR$.
   \end{itemize}
  The following is the main result of the paper.
   
   \begin{theorem}\label{t-main}
   	Given a range space $\cF$ satisfying (A1)-(A4) and $\eps,\delta,\omega\in(0,1)$, there is a deterministic (resp., randomized) algorithm that finds (resp., with probability $\Omega(1)$) a measure $\mu$ of support $K:=O(\frac{\gamma}{\eps^3(1-\omega)}\log \frac{\gamma}{\eps}\cdot\OPT_{\cF}\log\frac{\OPT_\cF}{\eps\delta(1-\omega)})$ that is a $(\frac{1+5\eps}{1-\omega},1-\delta)$-approximate solution for \raf{FH}, using $K$ calls to the oracle \MO$(\cF,w,\omega)$ (resp., \MO$(\cF,w,\sigma,\omega)$).	
   \end{theorem}   
                                
   In view of Corollary~\ref{cor1}, we have the following theorem as an immediate consequence of Theorem~\ref{t-main}.
   
   \begin{theorem}[Main Theorem]\label{t-main2}
   	Let $\cF=(Q,\cR)$ be a range space satisfying (A1)-(A4) and admitting a hitting set of size $s_{\cF}(\frac{1}{\epsilon})$ and $\eps,\delta,\omega\in(0,1)$ be given parameters. Then there is a (deterministic) algorithm that computes a set of size $s_{\cF}(z_{\cF}^*)$, hitting a subset of $\cR$ of measure at least $(1-\delta)w_0(\cR)$, using $O(\frac{\gamma}{\eps^3(1-\omega)}\log \frac{\gamma}{\eps}\cdot\OPT_{\cF}\log\frac{\OPT_\cF}{\eps\delta(1-\omega)})$ calls to the oracle \MO$(\ldots,\omega)$ and a single call to an $\epsilon$-net finder.
   \end{theorem}       
   In section \ref{sec:max}, we observe that the maximization oracle can be implemented in randomized polynomial time. As a consequence, we can extend Corollary~\ref{cor1} as follows (under the assumption of the availability of subsystem and sampling oracles in the dual range space); see Section~\ref{sec:max} for details.
    
   \begin{corollary}\label{cor-main} Let $\cF=(Q,\cR)$ be a range space of VC-dimension $d$ satisfying (A2) and (A3) and $\eps,\delta\in(0,1)$ be given parameters. Then there is a randomized algorithm that computes a set of size $O(d \cdot z_{\cF}^*\log(d\cdot z_{\cF}^*))$, hitting a subset of $\cR$ of measure at least $(1-\delta)w_0(\cR)$, in time $O(K\cdot g_{\cF^*}(\frac{d2^{d}\OPT_{\cF}^{2}}{\eps^2}\log \frac{\OPT_{\cF}}{\eps}))$, where $K:=O(\frac{d}{\eps^3(1-\eps)}\log \frac{d}{\eps}\cdot\OPT_{\cF}\log\frac{\OPT_\cF}{\eps\delta(1-\eps)})$.
   \end{corollary} 
   Note by Lemma~\ref{VC2} that $g_{\cF^*}(r)\le r^{2^{d+1}}$, but stronger bounds can be obtained for special cases. 

\section{The algorithm}   \label{sec:algorithm}
The algorithm is shown in Algorithm \ref{alg} below.
For any iteration $t$, let us define the {\it active} range-subspace $\cF_t=(Q,\cR_t)$ of $\cF$, where
\begin{align*}
\cR_t := \{ R \in \cR: |P_t\cap R| < T\}.
\end{align*}
Clearly, (since these properties are hereditary) $\VCd(\cF_t)\le\VCd(\cF)$, 
and $\cF_t$ admits and $\epsilon$-net of size $s_{\cF}(\frac{1}{\epsilon})$ whenever $\cF$ does. 
For convenience, we assume below that $P_t$ is (possibly) a multi-set (repetitions allowed).

  Define
  \begin{align} \label{eq:T}
  T_0&:=\frac{\OPT_{\cF}}{\eps(1-\omega)\delta^{1/\gamma}}\left(\ln\frac{1}{1-\eps}+\ln\frac{1}{\eps\delta}\right),~~ a:=\frac{\gamma}{\eps^2}, ~~\text{ and }b:=\max\{\ln T_0, 1\},\nonumber\\
  T&:=e^2 a b(\ln(a+e-1)+1)=\Theta(\frac{\gamma}{\eps^2}\log \frac{\gamma}{\eps}\log\frac{\OPT_\cF}{\eps\delta(1-\omega)}).
  \end{align} 
  For simplicity of presentation, we will assume in what follows that the maximization oracle is deterministic; the extension to the probabilistic case is straightforward. 
  
\setlength{\algomargin}{.25in}
\begin{algorithm}[H]
	\label{alg}
	\SetAlgoLined
	\KwData{A range space $\cF=(Q,\cR)$ satsfying (A1)-(A4), and an approximation accuracies $\eps,\delta,\omega,\in(0,1)$.}
	\KwResult{A $(\frac{1+5\eps}{1-\omega},1-\delta)$-approximate solution $\mu$ for \raf{FH}. }
   
    $t\gets 0$; $P_0\gets\emptyset$; set $T$ as in \raf{eq:T}\\
    \While{$w_0(\cR_t)\ge\delta\cdot w_0(\cR)$}{
       define the measure $w_t:\cR_t\to\RR_+$ by $ w_t(R) \gets (1-\eps)^{|P_t\cap R|}w_0(R)$, for $R\in\cR_t$\\	
       $p_{t+1}\gets\MO(\cF_t,w_t,\omega)$ \label{s-oracle}\\
       $P_{t+1}\gets P_{t}\cup\{p_{t+1}\}$\\
       $t \gets t+1$\\
    }
    \Return the measure $\widehat\mu:Q\to\RR_+$ defined by $\widehat\mu(q)\gets\frac{1}{T}\sum_{p\in P_t}\bdelta_p(q)$	
\caption{The fractional covering algorithm}
\end{algorithm}


\section{Analysis}\label{sec:analysis}

Define the potential function
$
\Phi(t) :=  w_t(\cR_t),
$
where $w_t(R) := (1-\eps)^{|P_t\cap R|}w_0(R)$, and $P_t=\{p_{t'}:t'=1,\ldots,t\}$ is the set of points selected by the algorithm in step~\ref{s-oracle} upto time $t$. We can also write $w_{t+1}(R) = w_t(R) (1-\eps\cdot\bbone_{p_{t+1}\in R})$.

The analysis is done in three steps: the first one (Section~\ref{sec:pot}), which is typical for MWU methods, is to bound the potential function, at each iteration, in terms of the ratio between the current solution obtained by the algorithm at that iteration and the optimum fractional solution. The second step (Section~\ref{sec:time}) is to bound the number of iterations until the desired fraction of the ranges is hit. Finally, the third step (Section~\ref{sec:convergence}) uses the previous two steps to show that the algorithm reaches the required accuracy after a polynomial number of iterations.
\subsection{Bounding the potential}\label{sec:pot}
The following three lemmas are obtained by the standard analysis of MWU methods with "$\sum$"$'$s replaced by "$\int$"$'$s. 
\begin{lemma} For all $t=0,1,\ldots$, it holds that
\begin{align}
\Phi(t+1) \leq \Phi(t) \exp\left(-\frac{\eps}{\Phi(t)}\cdot w_t(\cR_t[p_{t+1}])\right). \label{eq3}
\end{align}
\end{lemma}

\begin{proof}
	\begin{align*}
	\Phi(t+1) 	&=  \int_{R\in\cR_{t+1}} w_{t+1}(R) dR 
	=  \int_{R\in\cR_{t+1}} w_t(R) (1-\eps\cdot\bbone_{p_{t+1}\in R}) dR \nonumber \\
	&\leq  \int_{R\in\cR_{t}} w_t(R) (1-\eps\cdot\bbone_{p_{t+1}\in R}) dR  
    =  \Phi(t)  \left(1-\eps\int_{R\in\cR_{t}}\bbone_{p_{t+1}\in R}\frac{w_t(R)}{\Phi(t)} dR\right) \nonumber\\ 
	&\leq  \Phi(t) \exp\left(-\eps\int_{R\in\cR_{t}}\bbone_{p_{t+1}\in R}\frac{w_t(R)}{\Phi(t)} dR\right),
	\end{align*}
	where the first inequality is because $\cR_{t+1} \subseteq \cR_t$ since $|P_t \cap R|$ is non-decreasing in $t$, and the last inequality is because $1-z \leq e^{-z}$ for all $z$.
\end{proof}

\begin{lemma}
Let $\kappa(t):= \sum_{t'=0}^{t-1} \frac{w_{t'}(\cR_{t'}[p_{t'+1}])}{\Phi(t')}$. Then $z^*_{\cF}\cdot\kappa(t) \geq \frac{1-\omega}{1+\eps}|P(t)|$. \label{lem1}
\end{lemma}

\begin{proof}
Due to the choice of $p_{t'+1}$, we have that
\begin{align}
\xi_{t'}(p_{t'+1}):= w_{t'}(\cR_{t'}[p_{t'+1}])  &\geq (1-\omega)\max_{q \in Q} w_{t'}(\cR_{t'}[q]).  
\label{eq2}
\end{align}
Consequently, for a $(1+\epsilon)$-approximate solution $\mu^*$,
\begin{align*}
z_{\cF}^*\cdot\kappa(t) &= \sum_{t'=0}^{t-1}z^*_{\cF} \frac{w_{t'}(\cR_{t'}[p_{t'+1}])}{\Phi(t')} \ge\frac{1}{1+\eps} \sum_{t'=0}^{t-1} \left(\int_{q\in Q}\mu^*(q)dq\right)\int_{R\in\cR_{t'}} \bbone_{p_{t'+1}\in R}\frac{w_{t'}(R)}{\Phi(t')}dR \\
&\geq \frac{1-\omega}{1+\eps} \sum_{t'=0}^{t-1} \int_{q\in Q}\mu^*(q)dq \int_{R\in\cR_{t'}} \bbone_{q\in R}\frac{w_{t'}(R)}{\Phi(t')} dR
= \frac{1-\omega}{1+\eps} \sum_{t'=0}^{t-1} \int_{R\in\cR_{t'}}\left(\int_{q\in Q}\mu^*(q)\bbone_{q\in R}dq  \right)\frac{w_{t'}(R)}{\Phi(t')} dR \\
&= \frac{1-\omega}{1+\eps} \sum_{t'=0}^{t-1} \int_{R\in\cR_{t'}}\mu^*(R)  \frac{w_{t'}(R)}{\Phi(t')} dR 
\ge \frac{1-\omega}{1+\eps} \sum_{t'=0}^{t-1} \int_{R\in\cR_{t'}}  \frac{w_{t'}(R)}{\Phi(t')} dR \\
&= \frac{1-\omega}{1+\eps}\sum_{t'=0}^{t-1} 1=\frac{1-\omega}{1+\eps}|P(t)|.
\end{align*}
where the first inequality is due to the $(1+\eps)$-optimality of $\mu^*$, the second inequality is due to (\ref{eq2}), and the last inequality is due to the feasibility of $\mu^*$ for \raf{FH}.
\end{proof}

\begin{lemma}
	For all $t=0,1,\ldots,$ we have 
\begin{equation}\label{bd-pot}
\Phi(t) \leq \Phi(0) \exp\left( -{\eps}\cdot\frac{1-\omega}{1+\eps}\cdot\frac{|P(t)|}{z_{\cF}^*}\right).
\end{equation}
\end{lemma}
\begin{proof}
By repeated application of (\ref{eq3}), and using the result in Lemma \ref{lem1}, we can deduce that
\begin{align*}
\Phi(t) 	&\leq  \Phi(0) \exp\left( -\sum_{t'=0}^{t-1}\frac{\eps}{\Phi(t')}\cdot w_{t'}(\cR_{t'}[p_{t'+1}])\right)
= \Phi(0) \exp\left( -\eps\kappa(t)\right) \\
&\leq \Phi(0) \exp\left( -{\eps}\cdot\frac{1-\omega}{1+\eps}\cdot\frac{|P(t)|}{z_{\cF}^*}\right).
\end{align*}
\end{proof}

\subsection{Bounding the number of iterations}\label{sec:time}

\begin{lemma}\label{l-bd1} After at most $t_{\max}:=\frac{\OPT_{\cF}}{\eps(1-\omega)}\left(T\ln\frac{1}{1-\eps}+\ln\frac{1}{\eps\delta}\right)$ iterations, we have $w_0(\cR_{t_f})<\delta\cdot w_0(\cR)$.
\end{lemma}

\begin{proof}
For a range $R\in\cR$, let us denote by $\cT_t(R):=\{0\le t'\le t-1:~p_{t'+1}\in R\in\cR_{t'}\}$ the set of time steps, upto $t$, at which $R$ was hit by the selected point $p_{t'+1}$, when it was still active. Initialize $w'_0(R):=w_0(R)+\sum_{t'\in\cT_t(R)}w_{t'+1}(R).$ For the purpose of the analysis, we will think of the following update step during the algorithm: upon choosing $p_{t+1}$, set $w'_{t+1}(R):=w_t'(R)-w_t(R)\bbone_{p_{t+1}\in R}$ for all $R\in\cR_t$. Note that the above definition implies that $w'_t(R)\ge(1-\epsilon)^{|\cT_t(R)|}w_0(R)$ for all $R\in\cR$ and for all $t$.

\begin{claim}
	\label{cl1-1-1}
	For all $t$: 
	\begin{align}
	\label{e1-1-1}
	w'_{t+1}(\cR_{t+1})&\le \left(1-\frac{\eps(1-\omega)}{\OPT_{\cF}}\right)w'_{t}(\cR_t).
	\end{align}
\end{claim}
\begin{proof}
Consider an integral optimal solution $P^*\subseteq Q$ (which is guaranteed to exist by (A2)). Then 
	\begin{align}
	\label{e1-1-1-1}
	w_t(\cR_t)&=\int_{R\in\cR_t}w_t(R)dR=w_t\left(\bigcup_{q\in P^*}\cR_t[q]\right)\le \sum_{q\in P^*}w_t(\cR_t[q]).
	\end{align}
	From \raf{e1-1-1-1} it follows that there is a $q\in P^*$ such that  $w_t(\cR_t[q])\ge\frac{w_t(\cR_t)}{\OPT_{\cF}}.$ Note that for such $q$ we have 
	\begin{align}
	\label{e1-1-2}
	\xi_{t}(q)&:=w_t\left(\cR_t[q]\right)\ge \frac{w_t(\cR_t)}{\OPT_\cF}, 
	\end{align}
    and thus by the choice of $p_{t+1}$, $\xi_t(p_{t+1})\ge(1-\omega)\xi_t(q)\ge\frac{(1-\omega)w_t(\cR_t)}{\OPT_\cF}$. It follows that
    \begin{align}\label{e1-1-3}
    w'_{t+1}(\cR_{t+1})&\le w'_{t+1}(\cR_{t})=\int_{R\in\cR_t}(w_t'(R)-w_t(R)\bbone_{p_{t+1}\in R})dR\nonumber\\
    &=\int_{R\in\cR_t}w_t'(R)dR-\int_{R\in\cR_t}w_t(R)\bbone_{p_{t+1}\in R}dR=w_t'(\cR_t)-\xi_t(p_{t+1})\nonumber\\
    &\le w_t'(\cR_t)-\frac{(1-\omega)w_t(\cR_t)}{\OPT_\cF}.
    \end{align}
    Note that, for all $t$,
    \begin{align}\label{e1-1-4}
    w'_t(R)<w_t(R)\sum_{t'\ge 0}(1-\eps)^{t'}=\frac{w_t(R)}{\eps}. 
    \end{align}
    Thus, $w_t(\cR_t)>\eps\cdot w_t'(\cR_t)$.
    Using this in \raf{e1-1-3}, we get the claim.
\end{proof}
Claim~\ref{cl1-1-1} implies that, for $t=t_{\max}$,
 \begin{align*}
 w_t'(\cR_t)\le\left(1-\frac{\eps(1-\omega)}{\OPT_\cF}\right)^tw'_{0}(\cR)<e^{-\frac{\eps(1-\omega)}{\OPT_\cF}t}w'_{0}(\cR).
 \end{align*}
 Since $|R\cap P_t|<T$ for all $R\in \cR_t$, we have $w_t'(\cR_t)=\int_{R\in\cR_t}w_t'(R)dR>(1-\eps)^Tw_0(\cR_t)$. On the other hand, \raf{e1-1-4} implies that $w_0'(\cR)<\frac{ w_0(\cR)}{\eps}$. Thus, if $w_0(\cR_t)\ge\delta\cdot w_0(\cR)$, we get
 \begin{align*}\label{e1-1-6}
  (1-\eps)^T\delta<\frac{1}{\eps}\cdot e^{-\frac{\eps(1-\omega)}{\OPT_\cF}t},
 \end{align*}
  giving $t<\frac{\OPT_\cF}{\eps(1-\omega)}\left(T\ln\frac{1}{1-\eps}+\ln\frac{1}{\eps\delta}\right)=t_{\max}$, in contradiction to $t=t_{\max}$.
\end{proof}
\subsection{Convergence to an  $(\frac{1+4\eps}{1-\omega},1-\delta)$-approximate solution}\label{sec:convergence}

\begin{lemma}\label{l-bd2} Suppose that $T\ge\frac{\max\{1,\ln(g_\cF(t_{\max})/\delta)\}}{\eps^2}$ and $\eps \leq 0.68$. Then Algorithm \ref{alg} terminates with a $(\frac{(1+5\eps)}{1-\omega},1-\delta)$-approximate solution $\widehat\mu$ for \raf{FH}.
\end{lemma}

\begin{proof}
Suppose that Algorithm \ref{alg} (the while-loop) terminates in iteration $t_f \le t_{\max}$.
	
\noindent{\it $(1-\delta)$-Feasibility:} By the stopping criterion, $w_0(\cR_{t_f})<\delta\cdot w_0(\cR)$. Then for $t=t_f$ and any $R\in\cR\setminus\cR_t$, we have
$\widehat\mu(R)=\frac{1}{T}\int_{q\in R}\sum_{p\in P_t}\bdelta_p(q)dq=\frac{1}{T}\sum_{p\in P_t}\int_{q\in R}\bdelta_p(q)dq=\frac{1}{T}\sum_{p\in P_t}\bbone_{p\in R}=\frac{1}{T}|P_t\cap R|\ge 1$, since $|P_t\cap R|\ge T$, for all $R\in\cR\setminus\cR_t$.

\smallskip

\noindent{\it Quality of the solution $\widehat \mu$:} By assumption (A1), we have $|\cR_t|_{P_t}|\le g(|P_t|)$, for all $t$. Thus we can write

\begin{equation}\label{discrete}
\Phi(t)=\sum_{P\in\cR_t|_{P_t}}(1-\eps)^{|P|}w_0(\cR_t[P]),
\end{equation}
where $\cR_t[P]:=\{R\in\cR_t:~R\cap P_t= P\}$. Since $\Phi(t)$ satisfies \raf{bd-pot}, we get by \raf{discrete} that 
\begin{align*}
(1-\eps)^{|P|}w_0(\cR_t[P]) 
&\leq \Phi(0) \exp\left( -{\eps}\cdot\frac{1-\omega}{1+\eps}\cdot\frac{|P_t|}{z^*_{\cF}}\right), \quad \text{ for all }P\in\cR_t|_{P_t} \\
\therefore {|P|}\ln(1-\eps)+\ln(w_0(\cR_t[P])) &\leq \ln \Phi(0) - \eps\cdot\frac{1-\omega}{1+\eps}\cdot\frac{|P_t|}{z^*_{\cF}}, \quad \text{ for all }P\in\cR_t|_{P_t}.
\end{align*}
Dividing by $\eps\cdot\frac{1-\omega}{1+\eps} \cdot T$ and rearranging, we get
\begin{align}\label{e15b}
\frac{|P_t|}{z^*_{\cF}T}  \leq  \frac{(1+\eps)\left(\ln \Phi(0) - \ln(w_0(\cR_t[P])\right)}{\eps (1-\omega)T} + \frac{(1+\eps)|P|}{\eps (1-\omega)T}\cdot\ln\frac{1}{1-\eps}, \quad \text{ for all }P\in\cR_t|_{P_t}.
\end{align}
Since
\begin{align*}
 w_0(\cR_t)=w_0\left(\bigcup_{P\in\cR_t|_{P_t}}\cR_t[P]\right)= \sum_{P\in\cR_t|_{P_t}}w_0\left(\cR_t[P]\right),
\end{align*}
there is a set $\widehat P\in\cR_t|_{P_t}$ such that $w_0(\cR_t[\widehat P])\ge\frac{w_0(\cR_t)}{|\cR_t|_{P_t}|}$.

We apply \raf{e15b} for $t=t_f-1$ and $\widehat P\in\cR_t|_{P_t}$. 
Using $\Phi(0)=w_0(\cR)\le \frac{w_0(\cR_t)}{\delta}$, $|\cR|_{P_t}|\le g_\cF(|P_t|) \le g_\cF(t_{\max})$, $\widehat\mu(Q)=\frac{|P_t|+1}{T}$, $|\widehat P|< T$ (as $\widehat P=R\cap P_t$ for some $R\in\cR_t$), $T \geq \frac{\ln(g_\cF(t_{\max})/\delta)}{\eps^2}$ and $T\ge \frac{1}{\eps^2}$ (by assumption), and $z_{\cF^*}\ge 1$, we get
\begin{align*}
\frac{\widehat\mu(Q)}{z^*_{\cF}}  	&\le \frac{(1+\eps)\ln( g_\cF(t_{\max})/\delta)}{\eps (1-\omega)T} + \frac{(1+\eps)}{\eps (1-\omega)}\cdot\ln\frac{1}{1-\eps}+\frac{1}{T\cdot z_{\cF}^*}	\\
&\le \frac{\eps(1+\eps)}{(1-\omega)} + \frac{(1+\eps)}{\eps (1-\omega)}\cdot\ln\frac{1}{1-\eps}+\eps^2<\frac{1+5\eps}{1-\omega},	
\end{align*}
for $\eps \leq 0.68$. 
\end{proof}

\subsection{Satisfying the condition on $T$} As $g_\cF(t_{\max})\le t_{\max}^\gamma$, for some constant $\gamma\ge 1$ by assumption (A1), and $t_{\max}=\frac{ T\cdot\OPT_{\cF}}{\eps(1-\omega)}\left(\ln\frac{1}{1-\eps}+\ln\frac{1}{\eps\delta}\right)$ as defined in Lemma~\ref{l-bd1}, it is enough to select $T$ to satisfy $T\ge \frac{\gamma\ln T}{\eps^2}+\frac{\gamma\ln T_0}{\eps^2}$, or  
\begin{equation}\label{T-eq}
\frac{T}{a}>\ln {T}+b, 
\end{equation}
where $T_0$, $a$,  and $b$ are given by \raf{eq:T}.

Set $T=e^C a b(\ln(a+e-1)+1)$, where $C$ is a large enough constant. Then the left-hand side of \raf{T-eq} is $\frac{T}{a}=e^C b \ln(a+e-1)+ e^C b$, while the right-hand side is $\ln T+ b=C+\ln a+\ln (\ln(a+e-1)+1) +\ln b+ b$. Now we need to choose $C\ge 1$ such that $e^C>C+3$ (say $C=2$). Then
\begin{align*}
e^C b&>2b>b+\ln b,\\
e^C  b &\ln(a+e-1) 	\geq e^C \log(a+e-1) 
>(C+3) \ln(a+e-1) \\
&\ge C+1+2 \ln(a+e-1)
>C+ \ln a+\ln (\ln(a+e-1)+1).
\end{align*}
Thus, setting $T=\Theta(\frac{\gamma}{\eps^2}\log \frac{\gamma}{\eps}\log\frac{\OPT_\cF}{\eps\delta(1-\omega)})$ satisfies the required condition on $T$.                                                                                                                               

\section{Implementation of the maximization oracle}\label{sec:max}

Let $\cF=(Q,\cR)$ be a range space with $\VCd(\cF)=d$. Recall that the maximization oracle needs to find, for a given $\omega>0$ and measure $w:\cR\to\RR_+$, a point $p \in Q$ such that $\xi_w(p)\geq (1-\omega)\max_{q\in Q}\xi_w(q)$, where $\xi_w(p):=w(\cR[p])$.

We will assume here the availability of the following oracles:                   
\begin{itemize}
\item $\SO(\cF^*,\cR')$: this is the dual subsystem oracle; given a finite subset of ranges $\cR'\subseteq\cR$, it returns the set of ranges $\cQ_{\cR'}$. Note by Lemmas~\ref{VC} and~\ref{VC2} that $|\cQ(\cR')|\le g(|\cR'|,2^{d+1})$.
\item $\PO(\cF,\cR')$: Given $\cF$ and a finite subset of ranges $\cR'\subseteq\cR$, the oracle returns a point $p\in Q$ that lies in $\cap_{R\in\cR'} R$ (if one exists).
\item $\Sample(\cF,\widehat w)$: Given $\cF=(Q,\cR)$ and a probability measure $\hat w:\cR\to\RR_+$, it samples from $\widehat w$.   
\end{itemize}

To implement the maximization oracle, we follow the approach in \cite{CEH04}, based on $\epsilon$-approximations. Recall that an $\epsilon$-approximation for $\cF^*$ is a finite subset of ranges $\cR'\subseteq\cR$, such that \raf{eps-approx}
holds for all $q\in Q$. We use $\epsilon:=\frac{\omega}{2\OPT_{\cF}}$. By Theorem~\ref{thm:eps-approx}
a random sample $\cR'$ of size $N=O(\frac{d2^{d}}{\epsilon^2}\log\frac{1}{\epsilon})=O(\frac{d2^{d}\OPT_{\cF}^{2}}{\omega^2}\log \frac{\OPT_{\cF}}{\omega})$ from $\cR$ according to the probability measure $\widehat{w}:=w/w(\cR)$, is an $\epsilon$-approximation with high probability. We call $\SO(\cF^*,\cR')$ to  obtain the set $Q(\cR')$, then return the subset of ranges $\cR'''\in\argmax_{\cR''\in Q(\cR')}|\cR''|$. Finally, we call the oracle $\PO(\cF,\cR''')$ to obtain a point $p\in\cap_{R\in\cR'''}R$.

\begin{lemma}\label{l-max}
$\xi(p) \geq (1-\omega)\max_{q\in Q}\xi(q)$.
\end{lemma}

\begin{proof}
The proof, which we include for completeness, goes along the same lines in \cite{CEH04}. Let $q^*$ be a point in $\argmax_{q\in Q}\xi(q)$. Note that by assumption (A2), $w(\cR[q^*])\ge \frac{1}{\OPT_{\cF}}w(\cR)$. Then by \raf{eps-approx},
\begin{align*}
\frac{w(\cR[p])}{w(\cR)}&\ge \frac{|\cR'[p]|}{|\cR'|}-\epsilon\ge\frac{|\cR'[q^*]|}{|\cR'|}-\epsilon
\ge\frac{w(\cR[q^*])}{w(\cR)}-2\epsilon\\
&=\frac{w(\cR[q^*])}{w(\cR)}-\frac{w}{\OPT_{\cF}}
\ge(1-\omega)\frac{w(\cR[q^*])}{w(\cR)}.
\end{align*} 
The statement follows.
\end{proof}
\begin{remark}
	\label{r2}
	The above implementation of the maximization oracle assumes the \emph{unit-cost} model of computation and \emph{infinite} precision arithmetic (real RAM). In some of the applications in the next section, we note that, in fact, \emph{deterministic} algorithms exist for the maximization oracle, which can be implemented in the \emph{bit-model} with \emph{finite precision}. 
\end{remark}

\section{Applications}\label{sec:app}
\subsection{Art gallery problems}\label{sec:gallery}
In the art gallery problem we are given a (non-simple) polygon $H$ with $n$ vertices and  $h$ holes, and two sets of points $G,N\subseteq H$. Two points $p,q\in H$ are said to see each other, denoted by $p\sees q$, if the line segment joining them lies inside $H$ (say, including the boundary $\partial H$). The objective is to guard all the points in $N$ using candidate guards from $G$, that is, to find a subset $G'\subseteq G$ such that for every point $q\in N$, there is a point $p\in G'$ such that $p\sees q$.

Let $Q=G$, $\cR=\{V_H(q):~q\in N\}$, where $V_H(q):=\{p\in H:~p\sees q\}$ is the {\it visibility region} of $q\in H$. For convenience, we shall consider $\cR$ as a {\it multi-set} and hence assume that ranges in $\cR$ are in {\it one-to-one correspondence} with points in $N$.
We shall see below that the range space $\cF=(Q,\cR)$ satisfies (A1)-(A4). 
\paragraph{Related work.} Valtr \cite{V98} showed that $\VCd(\cF)\le 23$ for simple polygons and $\VCd(\cF)=O(\log h)$ for polygons with $h$ holes. For simple polygons, this has been improved to $\VCd(\cF)\le 14$ by Gilbers and Klein \cite{GK14}.                                                         

If one of the sets $G$ or $N$ is an (explicitly given) discrete set, then the problem can be easily reduced to a standard \textsc{SetCover} problem. For the case when $G=V$ is the vertex set of the polygon (called {\it vertex guards}), Ghosh \cite{G87, G10} gave an $O(\log n)$-approximation algorithm that runs in time $O(n^4)$ for simple polygons (resp., in time $O(n^5)$ for non-simple polygons). This has been improved by King \cite{K13} to $O(\log\log\OPT)$-approximation in time $O(n^3)$ for simple polygons (resp.,  $O((1+\log (h+1))\log\OPT)$-approximation in time $O(n^2v^3)$ for non-simple polygons), where $\OPT$ here is the size of an optimum set of {\it vertex} guards. The main ingredient for the improvement in the approximation ratio is the fact proved by King and Kirkpatrick \cite{KK11} that there is an $\epsilon$-net, in this case and in fact more generally when $G=\partial H$ (called {\it perimeter guards}), of size $O(\frac{1}{\epsilon}\log\log \frac{1}{\epsilon})$. 
                                                                                                                                 
For the case when both $N$ and $G$ are infinite, a {\it discretization} step, which selects a candidate discrete set of guards guaranteed to contain a near optimal-set, seems to be necessary for reducing the problem to \textsc{SetCover}. Such a discretization method was given in \cite{DKDS07} that allows an $O(\log\OPT)$-approximation for simple polygons (resp., $O(\log h\log(\OPT\cdot\log h))$-approximation in non-simple polygons) in {\it pseudo-polynomial} time $\poly(n,\Delta)$, where the {\it stretch} $\Delta$ is defined as the ratio between the longest and shortest distances between vertices of $H$. However, very recently, an error in one of the claims in \cite{DKDS07} was pointed out by Bonnet and Miltzow \cite{BM16},  who also suggested another discretization procedure that results in an $O(\log\OPT_{\cF})$-randomized approximation algorithm, after making the following two assumptions
\begin{itemize}
	\item[(AG1)] vertices of the polygon have integer components, given by their binary representation; 
	\item[(AG2)] no three extensions meet in a point that is not a vertex, where an extension is a line passing through two vertices.
\end{itemize}	                                                                                                                      
Under these assumptions, it was shown in \cite{BM16} that one can use a grid $\Gamma$ of cell size $\frac{1}{D^O(1)}$ such that $\OPT_\Gamma=O(\OPT_{\cF})$, where $\OPT_{\Gamma}$ is the size of an optimum set of guards {\it restricted} to $\Gamma$, and $D$ is the diameter of the polygon. Then one can use the algorithm suggested by Efrat and Har-Peled \cite{EH06} who gave an efficient implementation of the multiplicative weight updates method of Br\"{o}nnimann and Goodrich \cite{BG95} in the case when the locations of potential guards are restricted to a dense $\Gamma$. More precisely, the authors in \cite{EH06} gave a {\it randomized} $O(\log\OPT_{\Gamma})$-approximation algorithm for simple polygons (resp., $O(\log h\log(\OPT_{\Gamma}\cdot\log h))$-approximation in non-simple polygons) in expected time $O(n\OPT_{\Gamma}^2\log\OPT_{\Gamma}\log(n\OPT_{\Gamma})\log^2 \Delta)$ (resp.,  $nh\OPT_{\Gamma}^3\polylog n\log^2 \Delta)$), where $\Delta$ here denotes the ratio between the diameter of the polygon and the grid cell size. Note that this would imply a randomized (weakly) polynomial time approximation algorithm for the {\it unrestricted} guarding case, if one cand show that for a polygon with rational description (of its vertices), there is a near optimal set which has also a rational description. While it is not clear that this is the case in general, the main result in \cite{BM16} implies that $\Delta$ can be chosen, under assumption (AG2), to be $D^{O(1)}$, which implies by (AG1) that $\log \Delta$ is linear in the maximum bit-length of a vertex coordinate. Note that, the same argument combined with Theorem~\ref{BG} in the appendix shows that one can actually obtain $O(\log z_{\cF}^*)$-approximation in randomized polynomial-time for simple polygons under assumptions (AG1) and (AG2). 

On the hardness side \cite{ESW01}, the vertex (and point) guarding problem for simple polygons is known to be APX-hard \cite{ESW01}, while the problem for non-simple polygons is as hard as \textsc{SetCover} and hence cannot be approximated by a polynomial time algorithm with ratio less than $((1-\epsilon)/12)\log n)$, for any $\epsilon>0$, unless $\text{\it NP}\subseteq\text{\it TIME}(n^{O(\log\log n)})$.

\subsubsection{Point guards} \label{sec:pg}
In this case, we have $Q\leftrightarrow \cR\leftrightarrow G=N=H$. Note that (A1) is satisfied with $\gamma=\VCd(\cF)\le 14$ by Lemma~\ref{VC} and the result of \cite{GK14}.
It is also known (see, e.g., \cite{EH06}) that a subsystem oracle as in (A1$'$) can be computed efficiently, for any (finite) $P\subset Q$, as follows. Let $\cR':=\{V_H(p):~p\in P\}$. Then $\cR'$ is a finite set of polygons which induces an arrangement of lines (in $\RR^2$) of total complexity $O(nh|P|^2)$. We can construct this arrangement in time $O(nh|P|^2\log(nh|P|))$, and label each cell of the arrangement by the set of visibility polygons it is contained in. Then $\cR|_P$ is the set of different cell labels which can be obtained, for e.g., by a sweep algorithm in time $O(nh|P|^2\log(nh|P|))$.

(A2) follows immediately from the fact that each point in the polygon is seen from some vertex. 
(A3) is satisfied if we use $w_0\equiv 1$ to be the area measure over $H$ (recall that ranges in $\cR$ are in one-to-one correspondence with points in $H$). Thus we obtain the following result from Corollary~\ref{cor-main}, in the {\it unit-cost} model, since for any $\cR'\subseteq \cR$, $|Q|_{\cR'}|\le nh|\cR|^2$ and hence $g_{\cF^*}(r)\le nhr^2$.  

\begin{corollary}\label{cor3-}
	Given a polygon $H$ with $n$ vertices and $h$ holes and $\delta>0$, there is a randomized algorithm that finds in $O(nh^3\OPT_\cF^5\log\frac{\OPT_\cF}{\delta}\log^2\OPT_\cF\log^2(h+2))$ time a set of points in $H$ of size $O(z_\cF^*\log z_\cF^*\log(h+2))$ guarding at least $(1-\delta)$ of the area of $H$, where $z_\cF^*$ is the value of the optimal fractional solution. 
\end{corollary}

We obtain next a deterministic version of Corollary~\ref{cor3-} in the {\it bit model} of computation.
 
\paragraph{A deterministic maximization oracle.}~ 
We assume that the components of the vertices have rational representation, with maximum bit-length $L$ for each component  (i.e., essentially satisfy (AG1)). 
   
In a given iteration $t$ of Algorithm~\ref{alg}, we are given an active subset of ranges $\cR_t\subseteq\cR$, determined by the current set of chosen points $P_t\subseteq Q$, and the current measure $w_t:\cR_t\to\RR_+$, given by $ w_t(R)=(1-\eps)^{|P_t\cap R|}w_0(R)$, for $R\in\cR_t$, where $w_t(\cR_t)\ge\delta\cdot w_0(\cR)$. Let $\cR'_t:=\{V_H(p):~p\in P_t\}$. Note that, as explained above, the set of (convex) cells induced by $\cR_t'$ over $H$ has complexity (say number of edges) $r_t:=O(nh|P_t|^2)$ and can be computed in time\footnote{For simplicity of presentation, we do not attempt here to optimize the running time, for instance, by maintaining a data structure for computing $\cR_t'$, which can be efficiently updated when a new point is added to $P_t$.} $O(nh|P_t|^2\log(nh|P_t|))$; let us call this set $\cells(\cR_t')$, and for any $P\in\cells(\cR_t')$, define $\deg_t(P):=|\{p\in P_t:~p\sees q\text{ for some } q\in P\}|$ (recall that all points in $P$ are equivalent w.r.t. visibility from $P_t$). Note that (the subset of $H$ corresponding to) $\cR_t$ can be computed as $\cR_t=\bigcup_{P\in\cells(\cR_t'):\deg_t(P)<T}P$. We can write $\xi_t(q):=w_t(\cR_t[q])$ for any $q\in Q=H$ as
\begin{equation}\label{vis-wt}                                                                                                                     
\xi_t(q)=\sum_{P\in \cells(\cR_t')}(1-\eps)^{\deg_t(P)}\area(V_H(q)\cap P).
\end{equation}
Now, to find the point $q$ in $H$ maximizing $\xi_t(q)$, we follow\footnote{It should be noted that an FPTAS was claimed in \cite{NT94} when $w_t\equiv 1$, but this claim was not substantiated with a rigorous proof. In fact one of the statements leading to this claim does not seem to be correct, namely that the visibility region of the maximizer $q^*$ can be covered by a {\it constant} number of points that can be described only in terms of the input description of the polygon.} \cite{NT94} in expressing $\xi_t(q)$ as a (non-linear) continuous function of two variables, namely, the $x$ and $y$-coordinates of $q$. To do this, we first construct the partition $\{Q_1,\ldots,Q_l\}$ of $H$, induced by the arrangement of lines formed by the union $\cV$ of the vertices of $H$ and the vertices of $\cells(\cR_t')$. Note that for any convex cell $Q_i$ in this partition, any two points in $Q_i$ are equivalent w.r.t. the visibility of points from $\cV$. Moreover, for any pair of vertices $p,p'$ of a cell $P\in\cells(\cR_t')$, any two points in $Q_i$ lie on the same side of the line through $p$ and $p'$. This implies that, for any point $q=(x,y)\in Q_i\subseteq\RR^2$, the set $V_H(q)\cap P$
can be decomposed into at most $|E|=r_t$ regions that are either convex quadrilaterals or triangles, where $E$ is the set of edges of $\cells(\cR_t')$; see Figure~\ref{f1} for an illustration. Using the notation in the figure, we can write the vertices of the quadrilateral $Z$ in counterclockwise order as $q_i=(\frac{a_i(x,y)}{b_i(x,y)},\frac{c_i(x,y)}{d_i(x,y)})$, for $i=1,\ldots,4$, where $a_i(x,y),b_i(x,y), c_i(x,y)$, and $d_i(x,y)$ are affine functions of the form $Ax+By+C$, for some constants $A,B,C\in\QQ$ which are multi-linear of degree at most $3$ in the components of some of the vertices of $P$ and $H$. By the {\it Shoelace formula}, we can further write the area of $Z$ as   
\begin{equation}\label{areaZ}
\area(Z)=\frac{1}{2}\sum_{i=1}^4\frac{a_i(x,y)}{b_i(x,y)}\left(\frac{c_{i+1}(x,y)}{d_{i+1}(x,y)}-\frac{c_{i-1}(x,y)}{d_{i-1}(x,y)}\right),
\end{equation}
where indices wrap-around from $1$ to $4$. By considering a triangulation of $Q_i$, and letting $\Delta$ be the triangle containing $q\in Q_i$, we can write $q=(x,y)=\lambda_1(x',y')+\lambda_2(x'',y'')+(1-\lambda_1-\lambda_2)(x''',y''')$, where $(x',y'),(x',y')$ and $(x''',y''')$ are the vertices of $\Delta$, and $\lambda_1,\lambda_2\in[0,1]$.  
\begin{figure}	
	\centering  
\begin{subfigure}{.4\textwidth}
	\includegraphics[width=2.5in]{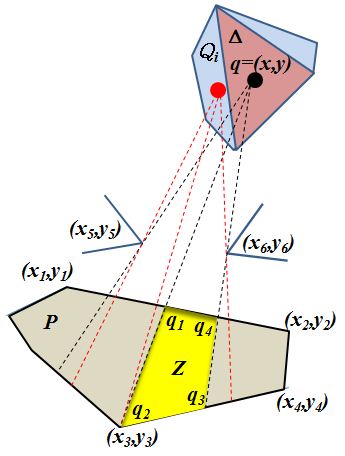}
	\caption{The visibility region of $q \in\Delta\subseteq Q_i$ inside $P$ can be decomposed into two convex quadrilaterals; one of them is $Z$, which is determined by $q=(x,y)$, the edges $\{(x_1,y_1),(x_2,y_2)\}$ and $\{(x_3,y_3),(x_4,y_4)\}$ of $P$, and the vertex $(x_6,y_6)$ of the polygon $H$. 
	}                                                                                                                                           
	\label{f1}
	\end{subfigure}
\begin{subfigure}{.4\textwidth}
	\includegraphics[width=2.5in]{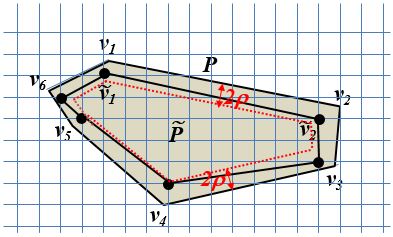}
\caption{The polygon $P$ and its internal approximation $\widetilde P$ using the vertices of the grid. Note that, since for all vertices $v$ of $P$ the distance between $v$ and $\widetilde v$ is at most $\rho$, the difference $P\setminus\widetilde P$ is covered by the region between the $\partial P$ and the dotted red polygon, which is at distance $\rho$ from the boundary of $P$} 
\label{f2}
\end{subfigure}
\hspace{0.5in}
\end{figure}
It follows from \raf{vis-wt} and \raf{areaZ} that $\xi_t(q)$ can be written as     \begin{equation}\label{vis-wt2}                                                                                                                                                         
\xi_t(q)=\xi_t(\lambda_1,\lambda_2)=\sum_{i=1}^k(1-\eps)^{j_i}\frac{M_i(\lambda_1,\lambda_2)}{N_i(\lambda_1,\lambda_2)},
\end{equation}     
where
 $k=O(|E|)=O(r_t)=\poly(n,h,\log\frac{1}{\delta})$, $j_i\le|P_t|\le t_{\max}=\poly(n,h,\log\frac{1}{\delta})$, and $M_i(\lambda_1,\lambda_2)$ and $N_i(\lambda_1,\lambda_2)$ are quadratic functions of $\lambda_1$ and $\lambda_2$ with coefficients having bit-length $O(L')$, where $L'$ is the maximum bit length needed to represent the components of the vertices in $\cV$. We can maximize $\xi_t(\lambda_1,\lambda_2)$ over $\lambda_1,\lambda_2\in[0,1]$
 by considering $9$ cases, corresponding to $\lambda_1\in\{0,1\}$, and $\lambda_1\in(0,1)$; and $\lambda_2\in\{0,1\}$, and $\lambda_2\in(0,1)$, and taking the value that maximizes $\xi_t(\lambda_1,\lambda_2)$ among them. Consider w.l.o.g. the case when $\lambda_1,\lambda_2\in(0,1)$. We can maximize $\xi_t(\lambda_1,\lambda_2)$ by setting the gradient of \raf{vis-wt2} to $0$, which in turn reduces to solving a system of two polynomial equations of degree $O(k)$ in two variables. A rational approximation to the solution $(\lambda_1^*,\lambda_2^*)$ of this system to within an additive accuracy of $\tau$ can be computed in time and bit complexity $\poly(L',k,\log\frac{1}{\tau})$, using, e.g., the {\em quantifier elimination algorithm} of Renegar \cite{R92}; see also Basu et al. \cite{BPR96} and Grigor'ev and Vorobjov \cite{GV88}. 
 \begin{claim}\label{cl4}
 	The function $\xi_t(\lambda_1,\lambda_2)$ in \raf{vis-wt2} is $ 2^{O(kL')}$-Lipschitz\footnote{A continuous differentiable function $f:S\to\RR$ is $\tau$-Lipschitz over $S\subseteq \RR^n$ if $|f(y)-f(x)|\le \tau\|x-y\|_2$ for all $x,y\in S$.}.                                                                                                                                                                                                           
 \end{claim}
 \begin{proof}              
It is enough to show that $\|\nabla\xi_t(\lambda_1,\lambda_2)\|_2\le 2^{O(kL')}.$ By \raf{vis-wt2}, each component of $\nabla\xi_t(\lambda_1,\lambda_2)$ is of the form $\frac{M(\lambda_1,\lambda_2)}{N(\lambda_1,\lambda_2)}$, where $M(\cdot,\cdot)$ and $N(\cdot,\cdot)$ are polynomials in $\lambda_1,\lambda_2\in[0,1]$ of degree $O(k)$ and coefficients of maximum bit length $O(kL')$. Thus $|M(\lambda_1,\lambda_2)|\le 2^{O(kL')}$. Also, from \raf{areaZ} and \raf{vis-wt2}, $N(\lambda_1,\lambda_2)$ can be written as a product of $k$ factors of the form $b(\lambda_1,\lambda_2)^2d(\lambda_1,\lambda_2)^2$, where $b(\cdot,\cdot)$ and $d(\cdot,\cdot)$ can be assumed to be {\it strictly positive} affine functions of $\lambda_1$ and $\lambda_2$. Suppose $b(\lambda_1,\lambda_2)=A\lambda_1+B\lambda_2+C$, for some constants $A,B,C\in\QQ$ which have bit length $O(L')$. Since the minimum of $b(\lambda_1,\lambda_2)$ over $\lambda_1,\lambda_2\in[0,1]$ is attained at some $\lambda_1,\lambda_2\in\{0,1\}$, it follows that $b(\lambda_1,\lambda_2)\ge\min\{C,A+C,B+C,A+B+C\}\ge\frac{1}{2^{O(L')}}$. A similar observation can be made for $d(\cdot,\cdot)$ and implies that $N(\lambda_1,\lambda_2)\ge\frac{1}{2^{O(kL')}}$, which in turn implies the claim.       	
 \end{proof}                                             
 Let $q^*_\Delta\in\argmax_{q\in\Delta}\xi_t(q)$. By the above claim, we can choose $\tau=\epsilon 2^{-O(kL')}$ sufficiently small, to get a point $q_\Delta\in\Delta$ such that $\xi_t(q_\Delta)\ge \xi_t(q^*_\Delta)- \epsilon$, where  $\epsilon:=\frac{\omega\cdot w_t(\cR_t)}{\OPT_{\cF}}\ge\frac{\omega \delta (1-\eps)^Tw_0(\cR)}{\OPT_{\cF}}$ (and hence $\log\frac{1}{\tau}=\poly(k,L',\log\frac{1}{\delta})$). Finally, we let $p\in\argmax_{\Delta} \xi_t(q_{\Delta})$, where $\Delta$ ranges over all triangles in the triangulations of $Q_1,\ldots, Q_l$, to get
 \begin{align*}
 \xi_t(p)\ge\max_{q\in Q}\xi_t(q)-\epsilon=\max_{q\in Q}\xi_t(q)-\frac{\omega\cdot w_t(\cR_t)}{\OPT_{\cF}}\ge(1-\omega)\max_{q\in Q}\xi_t(q),
 \end{align*}                                            
 where the last inequality follows from $\max_{q\in Q}\xi_t(q)\ge \frac{w_t(\cR_t)}{\OPT_{\cF}}$, implied by (A2).
                                                                                    
 \medskip
 
\noindent{\it\bf Rounding.}~ A technical hurdle in the above implementation of the maximization oracle is that the required bit length may grow from one iteration to the next (since the approximate maximizer $p$ above has bit length $\poly(k,L',\log\frac{1}{\tau})$), resulting in an exponential blow-up in the bit length needed for the computation.  To deal with this issue, we need to round the set $\cR_t$ in each iteration so that the total bit length in all iterations remains bounded by a polynomial in the input size\footnote{This is somewhat similar to the rounding step typically applied in numerical analysis to ensure that the intermediate numbers used during the computation have finite precision.}. This can be done as follows. Recall that $\cR_t$ can be decomposed by the current set of points $P_t$ into a set $\cells(\cR_t)$ of $r_t:=Cnh|P_t|^2$ disjoint convex polygons, for some constant $C>0$. Let $t_{\max}$ be the upper bound on the number of iterations given in Lemma~\ref{l-bd1}, and set $r_{\max}:=Cnht_{\max}^2$. We consider an infinite grid $\Gamma$ in the plane of cell size $\rho=\frac{\delta\cdot\area(\cR)}{16 D t_{\max} r_{\max}}$, where $D$ is the diameter of $H$ (which has bit length bounded by $O(L)$).

Let us call a cell $P\in\cells(\cR_t)$ {\it large} if $\area(P)\ge\frac{\delta\cdot\area(\cR)}{4r_{t}t_{\max}}$, and {\it small} otherwise. Let $\cL_t$ be the set of large cells in iteration $t$ of the algorithm. For each $P\in\cL_t$ we define an approximate polygon $\widetilde P\subseteq P$ as follows: for each vertex $v$ of $P$, we find a point $\widetilde v$ in $\Gamma\cap H$, closest to it, then define $\widetilde P:=\conv\{\widetilde v:~v \text{ is a vertex of $P$}\}$. Now, we let $\widetilde{\cR}_t:=\bigcup_{P\in\cL_t}\widetilde P$.
The following claim states that the total fraction of ranges that might not be covered due to this approximation is no more than $\delta/2$.
\begin{claim}\label{cl3}                                 
	$\sum_{t=1}^{t_f-1}\area(\cR_t\setminus\widetilde{\cR}_t)\le\frac{\delta}{2}\area(\cR)$.
\end{claim}
\begin{proof}
	Two sets contribute to the difference $\cR_t\setminus\widetilde{\cR}_t$: the set of small cells, and the truncated parts of the larges cells $\bigcup_{P\in\cL_t}P\setminus\widetilde{P}$.
	Note that the total area of the small cells is at most $ \sum_{t=1}^{t_f-1}r_t\cdot\frac{\delta\cdot\area(\cR)}{4r_{t}t_{\max}}<\frac{\delta}{4}\cdot\area(\cR)$. 
	On the other hand, for any $P\in\cL_t$, we have $\area(P)-\area(\widetilde P)\le 2\rho\cdot\prem(P)$, where $\prem(P)$ is the length of the perimeter of $P$. This inequality holds because $P\setminus\widetilde P$ is contained in the region at distance $2\rho$ from the boundary of $P$; see Figure \ref{f2} for an illustration. It follows that
	\begin{align*}
     \sum_{t=1}^{t_f-1}\sum_{P\in\cL_t}\area(P\setminus\widetilde P)\le 2\rho\cdot\sum_{t=1}^{t_f-1}\sum_{P\in\cL_t}\prem(P)\le 4\rho\cdot \sum_{t=1}^{t_f-1}r_t D < 4\rho t_f r_{t_f}D\leq\frac{\delta}{4}\cdot\area(\cR),
	\end{align*}
	by our selection of $\rho$. The claim follows.
\end{proof}
The only change we need in Algorithm~\ref{alg} is to replace $\cR_t$ in  by $\widetilde{\cR}_t$. (It is easy to see that the analysis also goes through with almost no change; we just have to replace $\cR_t$ by $\widetilde{\cR}_t$ and $\delta$ by $\frac{\delta}{2}$.)   

Note now that, since the polygon is contained in a square of size $2D$, the total number of points in $\Gamma$ we need to consider is at most $$\frac{2D}{\rho}=\frac{32 D^2t_{\max}r_{\max}}{\delta\cdot\area(H)}=2^{O(L)}\poly(n,h,\frac{1}{\delta}),$$ 
and thus the number of bits needed to represent each point of $\Gamma$ is $L\cdot\polylog(n,h,\frac{1}{\delta})$. 
Since the vertices of each cell $\widetilde{P}$ lie on the grid, the bit length $L'$ used in the computations above (in the implementation of the maximization oracle) and the overall running time is $\poly(L,n,h,\log\frac{1}{\delta})$.

\begin{corollary}\label{cor3}
Given a simple polygon $H$ with $n$ vertices with rational representation of maximum bit-length $L$ and $\delta>0$, there is a deterministic algorithm that finds in $\poly(L,n,\log\frac{1}{\delta})$ time a set of points in $H$ of size $O(z_\cF^*\log z_\cF^*\log(h+2))$ and bit complexity $\poly(L,n,\log\frac{1}{\delta})$ guarding at least $(1-\delta)$ of the area of $H$, where $z_\cF^*$ is the value of the optimal fractional solution. 
\end{corollary}
If $H$ is not simple, we get a result similar to Corollary~\ref{cor3} but with a quasi-polynomial running time $\poly(L,n^{O(\log h)},\log\frac{1}{\delta})$ (due to the complexity of the deterministic $\epsilon$-net finder).  
\begin{remark}\label{r3}
	It is worth noting that one can also obtain a \emph{ randomized} approximation algorithm with the same guarantee of Corollary~\ref{cor4} from the results in \cite{BM16}, by first randomly perturbing the polygon $H$ into a new polygon $H'$ such that $H'\subseteq H$ and $\area(H\setminus H') \le\delta$. Such a perturbation can be done using the rounding idea described above and guarantees with high probability that (AG2) is satisfied. Thus, we can apply the result in \cite{BM16} on $H'$.   
\end{remark} 
\subsubsection{Perimeter guards} 
In this case, we have $Q\leftrightarrow G=\partial H$ and $\cR\leftrightarrow T= H$. This is similar to the point guarding case with the exception that, in the maximization oracle, the point $q$ in \raf{vis-wt2} is selected from a line segment on $\partial H$. Also, by \cite{KK11}, the range space in this case admits an $\epsilon$-net of size $O(\frac{1}{\epsilon}\log\log \frac{1}{\epsilon})$. Thus we get the following result.

\begin{corollary}\label{cor4}
	Given a simple polygon $H$ with $n$ vertices with rational representation of maximum bit-length $L$ and $\delta>0$, there is a deterministic algorithm that finds in $\poly(L,n,h,\log\frac{1}{\delta})$ time a set of points in $\partial H$ of size $(z_\cF^*\log\log z_\cF^*)$ and bit complexity $\poly(L,n,h,\log\frac{1}{\delta})$ guarding at least $(1-\delta)$ of the area of $H$, where $z_\cF^*$ is the value of the optimal fractional solution. 
\end{corollary}

\subsection{Covering a polygonal region by translates of a convex polygon}\label{sec:cover-polygon}

Let $\cH$ be a collection of (non-simple) polygons in the plane and $H_0$ be a given {\it full-dimensional convex} polygon. The problem is to minimally cover all the points of the polygons in $\cH$ by translates of $H_0$, that is to find the minimum number of translates $H_0^1,\ldots,H_0^k$ of $H_0$ such that each point $p\in\bigcup_{H\in\cH}H$ is contained in some $H_0^i$. The discrete case when $\cH$ is a set of points has been considered extensively, e.g., covering points with unit disks/squares \cite{HM85} and generalizations in 3D \cite{ClarksonV06,Laue08}. 
Fewer results are known for the continuous case, e.g., \cite{G16} which considers the covering of simple polygons by translates of a rectangle\footnote{Note that in \cite{G16}, each polygon has to be covered {\it completely} by a rectangle.} and only provides an exact (exponential-time) algorithm; see also \cite{G11} for another example, where it is required to hit every polygon in $\cH$ by a copy of $H_0$ (but with rotations allowed).
  
This problem can be modeled as a hitting set problem in a range space $\cF=(Q,\cR)$, where $Q$ is the set of translates of $H_0$ and $\cR:=\left\{\{H_0^i\in Q:~R\in H_0^i\}:~R\in\bigcup_{H\in\cP}H\right\}$. Again considering $\cR$ as a multi-set, we have $\cR\leftrightarrow \bigcup_{H\in\cH}H$, and we shall refer to elements of $\cR$ as sets of translates of $H_0$ as well as points in $\bigcup_{H\in\cH}H$. It was shown by Pach and Woeginger \cite{PW90} that $\VCd(\cF^*)\le3$ and also that $\cF^*$ admits an $\epsilon$-net of size $s_{\cF^*}=O(\frac{1}{\epsilon})$. As observed in \cite{Laue08}, this would also imply that $\VCd(\cF)\le3$ and $s_{\cF}=O(\frac{1}{\epsilon})$. Thus (A1) is satisfied with $\gamma=3$; also we can show that (A1$'$) is satisfied as follows. Let $m$ be the total number of vertices of the polygons in $\cH$ and $H_0$. 
Given a finite subset $P\subseteq Q$ of translates of $H_0$, we can find (e.g. by a sweep line algorithm) in $O(m\log m)$ time the cells of the arrangement defined by $\cH\cup P$ (where a cell is naturally defined to be a maximal set of points in $\cR$ that all belong exactly to the same polygons in the arrangement). Let us cal this set $\cells(\cR)$ and note that it has size $O(m)$. Note also that every cell $\cR'\in\cells(\cR)$ is labeled by the subset $S(\cR')$ of $P$ that contains it, and $\cR|_P$ is the set of different labels.

Assume that $\cH$ is contained in a box of size $D$ and that $H_0$ contains a box of size $d$; then (A2) is satisfied as $\OPT_{\cF}\le\frac{D}{d}$. (A3) is satisfied if we use $w_0\equiv 1$ to be the area measure over $\cR$. Now we show that (A4) is also satisfied. 
 
Consider the randomized implementation of the maximization oracle in Section~\ref{sec:max}. We need to show that the oracles $\SO(\cF^*,\cR')$, $\PO(\cF,\cR')$ and $\Sample(\cF,w)$ can be implemented in polynomial time.
Note that for a given finite $\cR'\subseteq\cR$, the set $Q_{\cR'}$ is the set of all subsets of points in $\cR'$ that are contained in the same copy of $H_0$. Observe that each such subset is determined by at most two points from $\cR'$ that lie on the boundary of a copy of $H_0$. It follows that $\SO(\cF^*,\cR')$ can be implemented in $O((m|\cR'|)^2)$ time.  This argument also shows that $\PO(\cF,\cR')$ can be implemented in the time $\O((m|\cR'|)^2)$. Finally, we can implement $\Sample(\cF,\widehat w_t)$ given the probability measure $\widehat w_t:\cR\to\RR_+$ defined by the subset $P_t\subseteq Q$ as follows. We construct the cell arrangement $\cells(\cR)$, induced by $P=P_t$ as described above. 
We first sample $\cR'$ with probability $\frac{\widehat w_t(\cR')}{\sum_{\cR'\in\cells(\cR)}\widehat w_t(\cR')w_0(\cR')}$, then we sample a point $R$ uniformly at random from $\cR'$.  

\begin{corollary}\label{cor5}
Given a collection of polygons in the plane $\cH$ be and a (full-dimensional) convex polygon $H_0$, with $m$ total vertices respectively and $\delta>0$, there is a randomized algorithm that finds in $\poly(n,m,\log\frac{1}{\delta})$ time a set of $O(z_\cF^*)$ translates of $H_0$ covering at least $(1-\delta)$ of the total area of the polygons in $\cH$, where $z_\cF^*$ is the value of the optimal fractional solution. 
\end{corollary}

  
\subsection{Polyhedral separation in $\RR^d$}\label{sec:poly-sep}
Given two (full-dimensional) convex polytopes $\cP_1,\cP_2\subseteq \RR^d$ such that $\cP_1\subset \cP_2$, it is required to find a (separator) polytope $\cP_3\subseteq \RR^d$ such that $\cP_1\subseteq \cP_3\subseteq \cP_2$, with as few facets as possible. 
This problem can be modeled a hitting set problem in a range space $\cF=(Q,\cR)$, where $Q$ is the set of supporting hyperplanes for $P_1$ and $\cR:=\{\{p\in Q:~p\text{ separtaes $R$ from $\cP_1$}\}:~R\in\partial \cP_2\}$. Note that $\VCd(\cF)=d$ (and $\VCd(\cF^*)=d+1$). In their paper \cite{BG95}, Br\"{o}nnimann and Goodrich gave a deterministic $O(d^2\log\OPT_{\cF})$-approximation algorithm, improving on earlier results by Mitchell and Suri \cite{MS95}, and Clarkson \cite{C93}. It was shown in \cite{MS95} that, at the cost of losing a factor of $d$ in the approximation ratio, one can consider a finite set $Q$, consisting of the hyperplances passing through the facets of $\cP_1$. We can save this factor of $d$ by showing that $\cF$ satisfies (A1)-(A4). 

Let $n$ and $m$ be the number of facets of $\cP_1$ and $\cP_2$, respectively.  
Clearly (A1) is satisfied with $\gamma=d$, and given a finite set of hyperplanes $P\subseteq Q$ we can find the projection $\cR|_P$ as follows.
We first construct the cells of the hyperplane arrangement of $P$, which has complexity $O(|P|^d)$, in time $O(|P|^{d+1})$; see, e.g., \cite{AF92,S99}. Next, we intersect every facet of $\cP_2$ with every cell in the arrangement. This allows us to identify the partition of $\partial \cP_2$ induced by the cell arrangement; let us call it $\cells(\cR)$ (recall that $\cR\leftrightarrow\partial \cP_2$). Every $\cR'\in\cells(\cR)$ can be identified with the subset $S(\cR')$ of $P$ that separates a point $R\in\cR'$ from $\cP_1$. Then $\cR|_P=\{S(\cR'):~\cR'\in\cells(\cR)\}$. The running time for this is $\poly(|P|^d,m^d)$. Also, (A2) is obviously satisfied since $\cP_3=\cP_2$ is a separator with $n$ facets. For (A3), we use the $w_0\equiv 1$ to be the {\it surface area} measure (i.e., $w_0(\cR')=\vol_{d-1}(\cR')$ for $\cR'\subseteq\cR$). Now we show that (A4) also holds.

Consider the randomized implementation of the maximization oracle in Section~\ref{sec:max}. We need to show that the oracles $\SO(\cF^*,\cR')$, $\PO(\cF,\cR')$ and $\Sample(\cF,w)$ can be implemented in polynomial time.
Note that for a given finite $\cR'\subseteq\cR$, the set $Q_{\cR'}$ has size at most $g(|\cR'|,d+1)$, and furthermore, for any hyperplane $q\in Q$, $\cR'[q]$ is the set of points in $\cR'$ separated from $\cP_1$ by $q$. Thus, $\cR'[q]$ is determined by exactly $d$ points chosen from $\cR'$ and the vertices of $\cP_1$. It follows that the set $Q_{\cR'}$ can be found (and hence $\SO(\cF^*,\cR')$ can be implemented) in time $\poly((n^{\frac{d}{2}}+|\cR'|)^d)$. This argument also shows that $\PO(\cF,\cR')$ can be implemented in the time $\poly((n^{\frac{d}{2}}+|\cR'|)^d)$. Finally, we can implement $\Sample(\cF,\widehat w_t)$ given the probability measure $\widehat w_t:\cR\to\RR_+$ defined by the subset $P_t\subseteq Q$ as follows. We construct the cell arrangement $\cells(\cR)$, induced by $P=P_t$ as described above.  We first sample
$\cR'$ with probability $\frac{\widehat w_t(\cR')}{\sum_{\cR'\in\cells(\cR)}\widehat w_t(\cR')w_0(\cR')}$, then we sample a point $R$ uniformly at random from $\cR'$ (Note that both volume computation and uniform sampling can be done in polynomial time in fixed dimension). 
\begin{corollary}\label{cor6}
Given two convex polytopes $\cP_1,\cP_2\subseteq \RR^d$ such that $\cP_1\subset \cP_2$, with $n$ and $m$ facets respectively and $\delta>0$, there is a randomized algorithm that finds in $\poly((nm)^d,\log\frac{1}{\delta})$ time a polytope $\cP_3$ with $O(z_\cF^*\cdot d\log z_\cF^*)$ facets separating $\cP_1$ from a subset of $\partial\cP_2$ of volume at least $(1-\delta)$ of the volume of $\partial\cP_2$, where $z_\cF^*$ is the value of the optimal fractional solution. 
\end{corollary}

Note that the results in corollaries~\ref{cor5} and~\ref{cor6} assume the unit-cost model of computation and infinite precision arithmetic. We believe that deterministic algorithms for the maximization oracle in the bit-model can also be obtained using similar techniques as in Section \ref{sec:gallery}. We leave the details for the interested reader. 

\paragraph{Acknowledgement.}~The author is grateful to Waleed Najy for his help in the proof of Lemmas~\ref{l-bd2} and~\ref{l-max} and for many useful discussions.

\appendix

\section{An extension of the Br\"{o}nnimann-Goodrich algorithm for continuous range spaces}\label{sec:BG}
In addition to (A1), we will make the following assumption in this section:
\begin{itemize}
	\item[(A3$'$)] There exists a finite measure $\mu_0:P\to\RR_+$ such that the ranges in $\cR$ are $\mu_0$-measurable. 
\end{itemize}

\setlength{\algomargin}{.25in}
\begin{algorithm}[H]
	\label{BG-alg}
	\SetAlgoLined
	\KwData{A range space $\cF=(Q,\cR)$ satsfying (A1) and (A3$'$), and an approximation accuracy $\epsilon\in(0,1)$.}
	\KwResult{A hitting set for $\cR$. }
	
	$t\gets 0$\\
	\Repeat{$P$ is a hitting set for $\cR$}{
		define the probability measure $\widehat\mu_t:Q\to\RR_+$ by $\widehat\mu_t(q)\gets\frac{\mu_t(q)}{\mu_t(Q)}$, for $q\in Q$\\
		Find an $\epsilon$-net $P$ for $\cR$ w.r.t. the probability measure $\widehat\mu_t$ \label{BG-s1}\\
		\If{there is a range $R_{t+1}\in\cR$  such that $R_{t+1}\cap P=\emptyset$}{\label{BG-s2}
			$\mu_{t+1}(q):=2\mu_t(q)$ for all $q\in R_{t+1}$}
		$t \gets t+1$\\
	}
	\Return $P$	
	\caption{The Br\"{o}nnimann-Goodrich hitting set algorithm}
\end{algorithm}
For $\cR':=\{R_1,\ldots,R_t\}\subseteq\cR$, let $\cells(\cR):=\{\bigcap_{t'\in S}R_{t'}\setminus\bigcap_{t'\in [t]\setminus S}R_{t'}:~S\subseteq [t]\}\setminus\{\emptyset\}$ be the partition of $Q$ induced by $\cR'$. Define 
\begin{align}\label{delta0}
\delta_0:=\min_{\text{ finite } \cR'\subseteq\cR}\min_{P\in\cells(\cR')}\frac{\mu_0(P)}{\mu_0(Q)}.
\end{align}
\begin{theorem}\label{BG}
	Let $\cF=(Q,\cR)$ be a range space satisfying (A1) and (A3$'$) and admitting an $\epsilon$-net of size $s_{\cF}(\frac{1}{\epsilon})$, and $\mu:Q\to\RR_+$ be a measure feasible for \raf{FH}.
	For $\epsilon\le\frac{1}{2\mu(Q)}$, Algorithm~\ref{BG-alg} finds a hitting set  of size $s_{\cF}(\frac{1}{\epsilon})$ in $O(\mu(Q)\log\frac{1}{\delta_0})$ iterations.	
\end{theorem}
\begin{proof} Let $\cR_t=\{R_1,\ldots,R_t\}\subseteq \cR$ be the set of ranges whose weights are doubled in iterations $1,\ldots,t$. For $P\subseteq Q$, define $\deg_t(P):=|\{R\in\cR_t:~R\supseteq P\}|$.
	For two measures $\mu',\mu'':Q\to\RR_+$, denote by $\langle \mu',\mu''\rangle$ the inner product: 
	$\langle \mu',\mu''\rangle:=\int_{q\in Q}\mu'(q)\mu''(q)dq$.
	Then 
	\begin{equation}\label{e1-1}
	\langle \mu_t,\mu\rangle=\sum_{P\in\cells(\cR_t)}2^{\deg_t(P)}\mu(P).
	\end{equation}
	By the feasibility of $\mu$, for every $R_{t'} \in \cR_{t}$, we have that $\mu(R_{t'})=\int_{q\in Q}\mu(q)\bone_{q\in R_{t'}}dq \geq 1$. Thus,
	\begin{align}\label{e2-2}
	t&\le \sum_{t'=1}^t\mu(R_{t'})=\sum_{t'=1}^t\int_{q\in Q} \mu(q)\bone_{q\in R_{t'}}dq\nonumber\\
	& = \int_{q \in Q} \sum_{t'=1}^t\mu(q)\bone_{q\in R_{t'}}dq=\int_{q \in Q} \mu(q) \deg_t(q)dq\nonumber\\
	&=\sum_{P\in\cells(\cR_t)}\deg_t(P)\mu(P). 
	\end{align}
	From \raf{e1-1} and \raf{e2-2}, we obtain
	\begin{align}\label{e3-3}
	\frac{\langle \mu_t,\mu\rangle}{\mu(Q)}&=\sum_{P\in\cells(\cR_t)}2^{\deg_t(P)}\frac{\mu(P)}{\mu(Q)}
	\ge 2^{\sum_{P\in\cells(\cR_t)}\deg_t(P)\frac{\mu(P)}{\mu(Q)}}\ge 2^{t/\mu(Q)},
	\end{align}
	where the first inequality follows by the convexity of the exponential function while the second follows from \raf{e2-2}. 
	Since the range $R_{t+1}$ chosen in step~\ref{BG-s2} does not intersect the $\epsilon$-net $P$ chosen in step~\ref{BG-s1}, we have $\mu_t(R_{t+1})<\epsilon \mu_t(Q)$ and thus $\mu_{t+1}(Q)=\mu_t(Q)+\mu_t(R_{t+1})<(1+\epsilon)\mu_t(Q).$ 
	It follows that 
	\begin{align}\label{e4-4} 
	(1+\epsilon)^t>\frac{\mu_t(Q)}{\mu_0(Q)}&=\sum_{P\in\cells(\cR_t)}2^{\deg_t(P)}\frac{\mu_0(P)}{\mu_0(Q)}.
	\end{align}
	From \raf{e3-3}, we get that there is a $P\in\cells(\cR_t)$ such that $2^{\deg_t(P)}\ge 2^{t/\mu(Q)}$. On the other hand, \raf{e4-4} implies that $2^{\deg_t(P)}\frac{\mu_0(P)}{\mu_0(Q)}<(1+\epsilon)^t<e^{\epsilon t}$. Putting the two inequalities together, we obtain
	\begin{align}\label{e5-5}
	\frac{t\ln 2}{\mu(Q)}\le \epsilon t+\ln\frac{\mu_0(Q)}{\mu_0(P)}.
	\end{align}
	Since $\epsilon\le\frac{1}{2\mu(Q)}$, we get from \raf{e5-5} that $t\le \frac{1}{\ln 2-0.5}\cdot\mu(Q)\ln\frac{\mu_0(Q)}{\mu_0(P)}.$
\end{proof}
Let $\mu^*$ be a $(1+\eps)$-approximate solution for \raf{FH}. 
We can use Algorithm~\ref{BG-alg} in a binary search manner to determine whether or not $\mu^*(Q)\le(1+\rho)^i$, for any $\rho>0$ and $i\in\ZZ_+$, by checking if the algorithm stops with a hitting set in $\frac{1}{\ln 2-0.5}\cdot(1+\rho)^i\log\frac{1}{\delta_0}$ iterations. 
As $1\le \mu^*(Q)\le n$ if we assume (A2), we need only $O(\log_{1+\rho} n)$ binary search steps. 

We mention an application of Theorem~\ref{BG} when $Q$ is finite. Let $\mu_0\equiv 1$. Then $\delta_0\ge \frac{1}{|Q|}$, and Theorem~\ref{BG} implies that a hitting set of size $s_{\cF}(O(\mu^*(Q)))$ can be found in $O(\mu^*(Q)\log n\log |Q|)$ iterations.

\begin{remark}\label{r1}
One can also extend the second algorithm and analysis suggested in \cite{AP14} to the infinite case, to get as randomized algorithm that computes, with probability at least $\frac{1}{7}$ a hitting set of size $s_{\cF}(8z_{\cF}^*)$ in $O(z_{\cF}^*\ln(\frac{1}{\delta_0(\delta_0')^2}))$, where $\delta_0$ is as defined in \raf{delta0}, and 
\begin{align}\label{delta0'}
\delta_0':=\min_{\text{ finite } P\subseteq Q}\frac{w_0(\cR[P])}{w_0(\cR)},
\end{align}
where $\cR[P]:=\{R\in\cR:~R\cap Q= P\}$.
\end{remark}

\end{document}